%% file: main.tex
\documentclass[sigconf,authorversion,nonacm,table]{acmart}

\usepackage{booktabs} 
\usepackage{xspace}
\usepackage{amsmath,amsfonts}
\usepackage{algorithm}
\usepackage{algpseudocode}
\usepackage{graphicx}
\usepackage{textcomp}
\usepackage{tabularx}
\usepackage{makecell}
\usepackage{setspace}
\usepackage{enumitem} 
\usepackage{multirow}
    
\usepackage{orcidlink}
\def\BibTeX{{\rm B\kern-.05em{\sc i\kern-.025em b}\kern-.08em
    T\kern-.1667em\lower.7ex\hbox{E}\kern-.125emX}}

\newtheorem{claim}{Claim}

\newcommand\project{EdgeMiner\xspace}
\newcommand\optimization{MFP Requesting\xspace}

\newcolumntype{R}{>{\raggedleft\arraybackslash}X}
\newcolumntype{L}{>{\raggedright\arraybackslash}X}

\begin{document}
\title{EdgeMiner: Distributed Process Mining at the Data Sources}

\author{Julia Andersen}
\orcid{0009-0009-3305-4359}
\affiliation{%
  \institution{Kiel University}
  \country{Germany}
}
\email{jro@informatik.uni-kiel.de}

\author{Patrick Rathje}
\orcid{0000-0003-3718-7115}
\affiliation{%
  \institution{Kiel University}
  \country{Germany}
}
\email{pra@informatik.uni-kiel.de}

\author{Christian Imenkamp}
\orcid{0009-0007-4295-1268}
\affiliation{%
  \institution{University of Bayreuth}
  \country{Germany}
}
\email{christian.imenkamp@uni-bayreuth.de}

\author{Agnes Koschmider}
\orcid{0000-0001-8206-7636}
\affiliation{%
  \institution{University of Bayreuth}
  \country{Germany}
}
\email{agnes.koschmider@uni-bayreuth.de}

\author{Olaf Landsiedel}
\orcid{0000-0001-6432-300X}
\affiliation{%
  \institution{Kiel University}
  \country{Germany}}
\email{ol@informatik.uni-kiel.de}

\renewcommand{\shortauthors}{J. Andersen et al.}

\begin{abstract}
Process mining is moving beyond mining traditional event logs and nowadays includes, for example, data sourced from sensors in the Internet of Things (IoT). 
The volume and velocity of data generated by such sensors makes it increasingly challenging to efficiently process the data by traditional process discovery algorithms, which operate on a centralized event log. 
This paper presents \project, an algorithm for distributed process mining operating directly on sensor nodes on a stream of real-time event data. In contrast to centralized algorithms, \project tracks each event and its predecessor and successor events directly on the sensor node where the event is sensed and recorded. As \project aggregates direct successions on the individual nodes, the raw data does not need to be stored centrally, thus improving both scalability and privacy.
We analytically and experimentally show the correctness of \project.
In addition, our evaluation results show that \project determines predecessors for each event efficiently, reducing the communication overhead by up to 96\% compared to querying all nodes.
Further, we show that the number of queried nodes stabilizes after relatively few events, and batching predecessor queries in groups reduces the average queried nodes per event to less than 2.5\%. 
\end{abstract}


\maketitle

\input{edgeminer}

\end{document}

%% file: edgeminer.tex
\section{Introduction}  

Mining the never-ending data stream of events sensed by the Internet of Things (IoT) provides deep insights into, for example, processes in smart manufacturing, smart homes, and healthcare and is known as process mining~\cite{bertrand_bridging_model_pm_iot,frenandez_health_process_tracking,smart_factories,janssen2020process}.
However, already today, IoT devices generate significantly more data than we can store or transfer to the cloud~\cite{
GE2018601,sensingasaservice, vander_aalst_decomposing_2013}.
This high data velocity is a key challenge for process mining algorithms, which traditionally operate on static event logs, i.e., assume that data is readily available for mining at a central location with vast storage and computing capabilities~\cite{aalst2016datascienceinaction}.
Parsing such event logs, established algorithms commonly construct so-called footprint matrices (FMs), in which they store the order of events, so-called directly follows relations \cite{alpha_miner, leemans2013discovering, heuristicMiner}. 
Next to the scalability challenge, such collection of events at a central location causes privacy concerns \cite{alpha_miner_distributed,mannhardt2019privacy,pika2020privacy}.

\begin{figure}[tb]
    \centering
    \includegraphics[width=\linewidth]{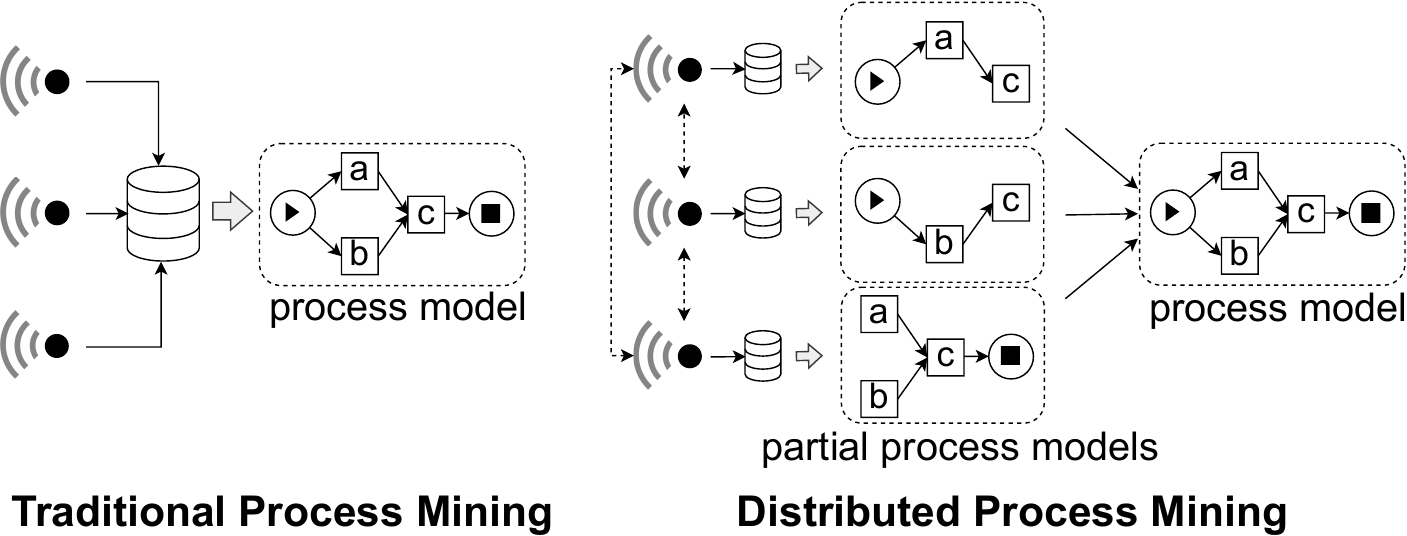}
    \caption{While traditional process mining (left) collects all events at a central entity, \project (right) processes them directly at the source, and only exchanges aggregates (partial footprint matrixes), increasing scalability and privacy.}
    \label{fig:centralizedToDistributed}
    \Description[From centralized to distributed process mining.]{While classical process mining (left) collects all events on a central entity, \project processes them in a distributed fashion directly at the source (right), and only exchanges data aggregates, increasing scalability and privacy.}
\end{figure}

In this paper, we introduce \project, an algorithm for distributed process mining: 
it composes FMs, the key building-block of process mining, in a distributed manner within a network of sensor nodes.
For each event that a node in \project detects, it locally stores the event as well as the predecessor and successor of this event and updates a local FM accordingly, see Fig.~\ref{fig:centralizedToDistributed}.
Thus, over time, each node constructs a partial FM that represents the view of this node on the direct successions of this process.
To derive the complete FM of the process, one of possibly multiple (central) entities queries the nodes for their partial FMs and combines these into a complete matrix.
As a result, \project enables (a) scalable mining, as a node only interacts with its predecessors after each event and, when queried, only exchanges aggregates, i.e., partial FMs, 
with the requesting entity and (b) privacy-preserving process mining, as nodes locally merely store their own as well as predecessor and successor events and a requesting entity only receives aggregate FMs instead of raw event streams. 
We prove, in this paper, that \project guarantees an FM identical to the FM when using a centralized approach.

Analyzing real-world datasets, our evaluation shows that activities statistically only have a few predecessor activities, which further only rarely change throughout the lifetime of a process. 
For example, even in the large Hospital Log~\cite{hospital}, comprising more than 600 distinct activities, the average count of predecessors per activity remains modest, averaging fewer than 7 predecessors.
As a result, for each event, a node can very efficiently determine the predecessor, reducing the communication overhead by up to 96\% compared to the naive baseline of querying all nodes in the network. 
Further, by batching predecessor queries in batches of, for example, 40 events, we show that \project can reduce the average number of queried nodes per event to less than 2.5\% of all nodes.
As a final result, we show that for most datasets, an intermediate FM already complies with the final FM by more than 90\% after less than 5\% of the events of the dataset.

In summary, the contributions of this paper are as follows.
\begin{itemize}
    \item {We introduce \project, a distributed process mining algorithm enabling distributed process discovery and conformance checking by operating directly where the events originate.}
    \item We ensure that the nodes only share aggregated event data with a requesting entity. 
    \item {We analytically and experimentally show the correctness of the merge of \project.}
    \item {We demonstrate that \project's interim results quickly converge to the footprint matrix of the whole event log.}
    \item We show that \project reduces the message count for the predecessor determination by up to 96\% compared to the naive approach and, for example, by 97.5\% when batching queries with a batch size of 40.
    \item We implement \project and publish its code open-source\footnote{Available at: \url{https://github.com/ds-kiel/EdgeMiner}}. 
\end{itemize}

The remainder of the paper is structured as follows.
Section~\ref{sec:Background} introduces the necessary background.
Next, Section~\ref{sec:RelatedWork} presents related work on distributed process mining.
Afterwards, we introduce the design of \project in Section~\ref{sec:Design}.
In Section~\ref{sec:Evaluation}, we evaluate \project.
We conclude the paper in Section~\ref{sec:Conclusion}.

\section{Preliminaries and Background}
\label{sec:Background}

In this section, we give a brief introduction into process mining and provide definitions of events, traces, event logs, direct successions and footprint matrices.
In process mining, we commonly discover a process model from an event log (process discovery) or check that a process follows a given model (conformance checking) \cite{process2011vanderAalst,aalst2016datascienceinaction,van_der_aalst_process_2012,van_der_aalst_process_2012-1}.
A common way to represent a process model is a Petri net \cite{peterson_petri_1977,murata_petri_1989}.

\textbf{Events and Event Logs: }
Process mining algorithms operate on event logs.
Each event in the log typically contains at least three attributes: an activity, a timestamp, and a case ID.
The case ID is used to distinguish different process instances.
A trace is a sequence of events corresponding to a particular case ordered by timestamp.

\textbf{Direct Successions: }
In the following, let $L$ be an event log over the activity set $T$.
We define direct successions and therewith definable relationships between two activities, following the definition by Gatta et al.~\cite{alpha_miner_distributed}.
%
\begin{definition}[Direct Succession]
\label{def:directsuccession}
    Activity $b$ directly succeeds $a$ in event log $L$, i.e. $a >_L b$, if and only if there exists a trace $\sigma = (t_0,\dots,t_{n-1})\in L$ and $j \in \{0,\dots,n-2\}$ such that $t_j = a ~\land~ t_{j+1} = b$. 
\end{definition}
%
\begin{definition}[Causality, No Direct Succession]
\label{def:symbols}
    With the help of Definition~\ref{def:directsuccession} we now define causality (\ref{eq:causality}) 
    and no direct succession (\ref{eq:nosuccession}) between two activities.
    \begin{align}
        a \rightarrow_L &~b   \!\!\!\!\!\!\!\!\!\!\!\!\!\!      &\Leftrightarrow    \quad ~a >_L b ~~\land~~ b\not>_L a \label{eq:causality}\\
        a ~\#_L &~b           \!\!\!\!\!\!\!\!\!\!\!\!\!\!       &\Leftrightarrow    \quad ~a \not >_L b ~~\land~~ b \not>_L a \label{eq:nosuccession}
    \end{align}
\end{definition}

\textbf{Footprint Matrix: }
A footprint matrix (FM) stores information regarding the direct succession relationships between activities within the event log.
For event log $L$, the corresponding footprint matrix $FM$ has dimension $\vert T_L \vert \times \vert T_L \vert$, where each activity in $L$ is mapped to a specific row and column in the matrix, i.e., we can enumerate the activities such that each activity number $i$ is an element of $\{0, \dots, \vert T_L\vert -1\}$.
An entry denotes how often the activity corresponding to the row was followed by the activity of the column.
This way, $FM_{ij}\in \mathbb{N}_0$ denotes the entry of the footprint matrix that describes the relationship between activities $i$ and $j$.
\begin{align*}
    FM_{ij} > 0 \quad &\Leftrightarrow \quad i >_L j \\
    FM_{ij} = 0 \quad &\Leftrightarrow \quad i \not>_L j
\end{align*}
For our design, we define partial footprint matrices containing only one specific column: 
For $i,j \in \{0, \dots, \vert T_L\vert -1\}$, we define $FM^i$ as the footprint matrix that only contains column $i$ and $FM^i_j$ as the entry of column $i$ in row $j$.

%

\textbf{Alpha Miner: }
To illustrate process mining and its use of FMs, we recapitulate the main steps of the seminal process discovery algorithm Alpha Miner (Algorithm~\ref{alg:alphaminer}). Later, we discuss how newer discovery algorithms extend the concept of the alpha miner and its FMs.

First, the algorithm denotes all activities occurring in the event log $L$ as $T_L$. 
Each activity of this set corresponds to a transition in the final output model computed by the algorithm. 
Then, the algorithm determines two sets, $T_I$ and $T_O$, which contain the start activities of the traces in $L$ and their end activities, respectively.
Subsequently, the algorithm determines all direct successions in $L$ which it stores in an FM.
This FM is the basis for the further calculations in the algorithm and the resulting Petri net.
Next, the algorithm generates a set $X_L$, which comprises pairs of activity sets. Each element of the first set is directly followed by every element of the second set, although not vice versa. 
Moreover, activities within each set do not directly succeed one another. 
The Alpha Miner minimizes the set of pairs, followed by the derivation of Petri net places from the minimized pairs in the next step. 
Finally, the algorithm adds arcs between these places and transitions, forming the resulting Petri net.

\begin{algorithm} \onehalfspacing
\caption{Alpha Miner}\label{alg:alphaminer}
\begin{algorithmic}[1]
    \State $T_L  = \{t\in T | \exists \sigma \in L: t \in \sigma\}$ 
    \State $T_I = \{first(\sigma)\in T | \forall \sigma \in L\}$
    \State $T_O = \{last(\sigma)\in T | \forall \sigma \in L\}$
    \State $\begin{aligned}[t]
        X_L = &\{(A,B) | A \subseteq T_L \land B \subseteq T_L \\[-3pt]
        & \land \forall a \in A \forall b \in B: a \rightarrow_L b \\[-3pt]
        &\land \forall a_1,a_2 \in A: a_1 \#_L a_2 \land \forall b_1,b_2 \in B: b_1 \#_L b_2\}\\[2pt]
        \end{aligned}$
    \State $\begin{aligned}[t]
        Y_L = &\{(A,B) \in X_L | \forall (A',B') \in X_L: A\subseteq A' \land B\subseteq B'\\[-3pt]
        &\Rightarrow (A,B) = (A',B')\} \\[2pt]
        \end{aligned}$
    \State $P_L = \{ p_{(A,B)} | (A,B) \in Y_L\} \cup \{i_L, o_L\}$
    \State $\begin{aligned}[t]
        F_L = 
        &\{(a, p_{(A,B)}) | (A,B) \in Y_L \land a \in A \} \\[-3pt]
        &\cup \{ (p_{(A,B)},b) | (A,B) \in Y_L\land b \in B \}  \\[-3pt]
        &\cup \{(i_L,t) | t \in T_I\} \cup \{(t,o_L) | t \in T_O\}\\[2pt]
        \end{aligned}$
    \State $\alpha (L) = (P_L, T_L, F_L)$
\end{algorithmic}
\end{algorithm}

To form the set $X_L$, we need to have knowledge of the $\rightarrow_L$ and $\#_L$ relationships between activities.
These relations solely depend on the direct successions between the activities.
An FM stores exactly those direct successions. 
For the Alpha Miner, it is not necessary to know the number of direct successions, just whether they exist or not.

\textbf{Heuristics Miner and Inductive Miner:}
Newer process discovery algorithms, extent on the concept of the Alpha Miner while keeping the FM as a key building block.
For example, the Heuristics Miner starts with the construction of a dependency graph \cite{heuristicMiner}.
To build the dependency graph, it calculates a so-called dependency measure between each two activities, depending on the number of direct successions between the activities.
As an FM stores all direct successions, the miner use it for this calculation.
The Inductive Miner also relies on direct successions stored in an FM. 
The core concept is to construct a directly-follows graph, which visually represents the FM along with start and end activity sets. 
From this graph, a process tree is formed by repeatedly partitioning the nodes of the directly-follows graph, capturing the structure of the process step by step.

\section{Related Work}
\label{sec:RelatedWork}


Today's approaches to distributed process discovery focus primarily on multiple compute nodes to speed up the processing of ever-growing, centralized event logs:
Van der Aalst~\cite{vander_aalst_decomposing_2013} presents one of the first methods for distributed process discovery.
He introduces a method for computing partial and overlapping process models in the form of Petri nets from partial event logs and defines a way to merge them to a Petri net of the complete event log. In contrast to \project, he creates models rather than FMs in a distributed manner. 
Others~\cite{alpha_miner_map_reduce} leverage Map-Reduce~\cite{dean2008mapreduce} to parallelize the Alpha Miner, distributing the computational workload across multiple compute nodes to 
handle the ever-growing size of event logs.
However, their approach still relies on an event log that is initially stored at a central location. 
Later works ~\cite{evermann_scalable_2016} bring Map-Reduce to the Flexible Heuristics Miner~\cite{flexibleHeuristicsMiner}.
Gatta et al.~\cite{alpha_miner_distributed} extends the primary motivation of the previous work (i.e., performance and scalability) to include privacy concerns in the healthcare domain, but still focuses on one event log per hospital.
In contrast to all these, \project emphasizes scalability, i.e., processing data directly at the source and privacy preservation by only sending aggregate information, i.e., partial FMs, to the requesting node. 

In the context of conference checking, the Single-Entry Single-Exit (SESE) breaks down large process models and event logs into smaller sub-processes to analyze them independently~\cite{MUNOZGAMA2014102}.
Similarly, passage-based decomposition~\cite{van2014process} and Projected Conformance Checking (PCC)~\cite{leemans2018scalable}
split logs into sets of activities that share direct successor-predecessor relationships: 
By solving each set independently and merging the results, it distributes, for example, conformance checking across multiple nodes.
Although these approaches efficiently handle large and complex event logs, they remain centralized. 
One recent work presents conformance checking in a distributed manner also based on footprint-matrices, yet -- in contrast to \project -- their work does not cover the underlying mechanism to distribute the calculation of the matrices~\cite{andersen2024discc}.

Streaming process mining~\cite{Burattin2022} takes a different approach to distributed process mining. It focuses mainly on the \emph{real-time} analysis of streams of events (i.e., event streams). 
Generally, those algorithms maintain computational state according to on previous events~\cite{6900341, 7376771, DBLP:journals/ijdsa/ZelstBHDA19, DBLP:journals/corr/abs-1212-6383} and employ approximations such as lossy counting~\cite{6900341}.
While those methods are also motivated by scalability and performance, they are not designed to compute FMs at the sensor nodes.

In contrast, \project distributes the FM computation and each sensor node computes a partial FM, enabling data aggregation at an earlier stage. 
As a result of processing events directly at their sources, \project increases scalability, reduces privacy risks and enables both process discovery and online conformance checking.

\section{Design}
\label{sec:Design}

In this section, we begin by discussing our assumptions and setting. Next, we give a design overview of \project, detail on its design, introduce optimizations, prove correctness, and 
discuss how \project enables distribution of well-known process mining algorithms. 

\subsection{Assumptions and Setting}

\subsubsection{Assumptions on Events and Event Logs}

In adherence to common practice, we assume events with the attributes case ID, activity name, and timestamp, whereas the case ID is unique for each case.
We assume timestamps within the same case to be distinct, ensuring unique FMs. 
In our evaluation in Section~\ref{sec:Evaluation}, we show that our assumptions are both realistic and applicable in numerous real-world scenarios.

\subsubsection{\project Setting}

\project targets typical IoT application settings: 
A network of IoT nodes that are, for example, integrated into machines in a smart factory or smart health setting. 
These sensors can communicate with each other, either directly or via the Internet.
Further, we assume that a membership management allows nodes to gain knowledge of other nodes and how to reach them, as common in IoT applications.
Finally, we assume that their clocks are well synchronized, using, for example, established clock-synchronization algorithms such as NTP~\cite{mills1985network}. 
This is essential, as accurate timestamps ensure proper event ordering. 
%
Without loss of generality, and to ease simplicity and comprehensibility, we assume each node detects exactly one type of activity.
Note that by introducing virtual nodes, we can easily extend the problem setting to allow each node to be responsible for multiple activity types.

\subsection{\project Algorithm}

\project consists of two phases:
Phase 1 operates continuously upon each event detected, and Phase 2 is executed upon a footprint matrix request.
In Phase 1, nodes communicate to determine predecessors when recording an event and note their directly-follows relations in their local FMs. 
As a result, every node continuously updates its local, partial FM.
Other than Phase 1, Phase 2 runs on demand, whenever a (central) entity requests the FM.
Then, all nodes send their partial FMs to the requesting entity, which constructs the FM.

\subsubsection{Phase 1 -- Event Ordering \& Partial Footprint Matrices}

\begin{figure*}[tb]
    \centering
    \includegraphics[width=1\linewidth]{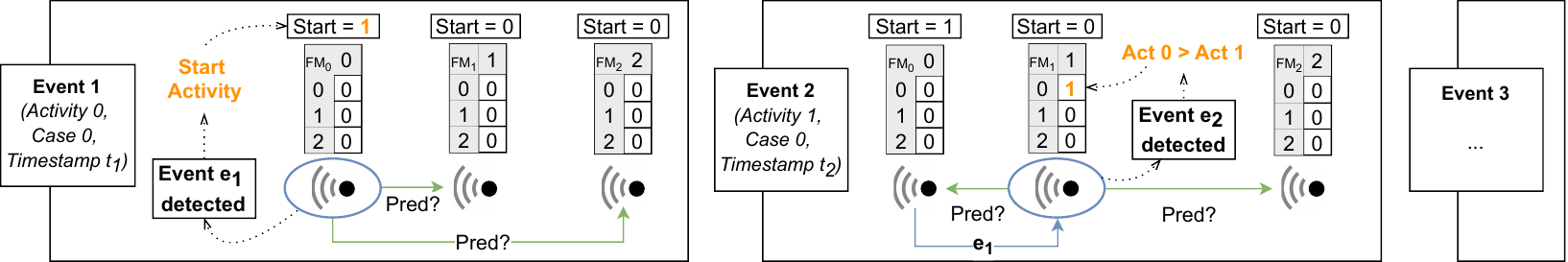}
    \caption{Phase 1 -- Event Ordering and Partial Footprint Matrix Construction in \project: Without a central entity, nodes determine the order of events collaboratively using message passing.
    In this example, event 1 is a start event. Therefore, after the node detects the event, it queries the other nodes for the predecessor event. In our example it does not get any positive responses, and, thus, denotes it detected a start event. Upon sensing event 2, the detecting node queries for the predecessor event, receives a response, listing event 1 as predecessor. It stores this information in its local FM.}
    \label{fig:toyExample}
    \Description[Phase 1 -- Event Ordering and Partial Footprint Matrix Construction in \project]{Phase 1 -- Event Ordering and Partial Footprint Matrix Construction in \project: Without a central entity, nodes determine the order of events collaboratively using message passing.
    In this specific setting, event 1 is a start event. Therefore, after the node detects the event, it queries for the predecessor event, for which it does not get any responses. Thus, the node sets a flag that it detected a start event. For event 2, after sensing the event and querying for the predecessor event, the respective node receives a response, including event 1. It increments the value in the corresponding cell of its local FM.}
\end{figure*}

For simplicity, we begin by introducing a naive version of \project, which is not communication efficient. 
Later, we introduce optimizations that ensure efficiency. 
When an event $e$ on node $i$ is triggered, node $i$ executes the following steps; see Fig.~\ref{fig:toyExample} for an example.
First, node $i$ stores event $e$ locally, including its case ID, the activity name and timestamp.
Then, it sends all other nodes a message containing $e$ to determine the predecessor event. 

When receiving the request from $i$, a node checks its storage for events in the same case that have no successor yet and whose timestamps are smaller than the timestamp of $e$.
If such an event exists, the node responds with the potential predecessor event data. 
If there are multiple responses, node $i$ picks the event with the latest timestamp that is still smaller than the own one, as predecessor. 
We denote the node with the predecessor event as node $j$. 
Node $i$ stores the newly determined predecessor locally, updates its local footprint matrix $FM^i$ by increasing the count in the respective cell, and informs $j$.
Subsequently, node $j$ saves $e$ as the successor event of its own event.
Events which do not have successors for a long duration are flagged as end activities. 

If a node has to choose its predecessor from several potential predecessor events, this indicates that events are happening in quick succession and (a) communication is delayed or (b) multiple nodes are identifying predecessors in parallel.
This may lead to temporarily incorrectly set predecessors.
\project automatically rectifies such inconsistencies during the run of the algorithm as soon as the node that belongs to the incorrectly set predecessor event receives a request from the node of its correct successor: The node sets the correct successor and informs the incorrect successor.
Then, the incorrect one starts the predecessor search again. 
A node then corrects the counters in its local FM.
The rectification process converges since each trace has a single start event, and because the nodes compare the timestamps of set successor events with the timestamps of incoming predecessor-search-requests.

Formally, each node $k \in \{0,\dots, n-1 \}$ only keeps track of its predecessor activities in its FM.
Therefore, each local FM has dimension $n \times 1$, with $n$ the number of activities.
The value at index $\ell \in \{0,\dots, n-1 \}$ describes how often $k$ directly follows $\ell$ in all of the cases, formally $\ell>_Lk$.
It is necessary to increment $FM^k_\ell$ by $1$ each time $k$ detects $\ell>_Lk$ and subtract $1$ if a direct succession gets fixed.
If a node does not find a predecessor, it determines that it recorded the start of a trace and logs the event as a start event. 
In \project, a node searches for its predecessor rather than its successor since it can be sure that the preceding event already occurred and, thus, the node can search for it immediately.

\subsubsection{Phase 2 -- Requesting a Footprint Matrix}
This phase begins when an outside entity requests the FM.
During this phase, the outside entity requests the local FMs of all nodes and start and end events within any case, see Fig.~\ref{fig:edgeminerphase2}.
Subsequently, each node responds with its local FM.
The requesting entity receives all responses, and it assembles the footprint matrix $FM^{all}\in \mathbb{N}_0^{n\times n}$ by concatenating the columns $FM^i$, $i\in\{0,\dots,n-1\}$ of all $n$ nodes, ordered by activity ID. 
Next, this FM serves as a process model for conformance checking, which compares the next events, traces, or sets of traces. 
Further, process discovery algorithms such as the Alpha Miner~\cite{alpha_miner}, the Heuristics Miner~\cite{heuristicMiner}, or the Inductive Miner~\cite{leemans2013discovering} can use this assembled FM for process discovery.
Please note that the outside entity only receives aggregate information, i.e., the partial FMs. 
Raw event data is only processed on the sensor nodes, ensuring both scalability and data privacy.

\begin{figure}[tb]
    \centering
    \includegraphics[width=1\linewidth,trim=0.7cm 0.5cm 0.9cm 0.5cm]{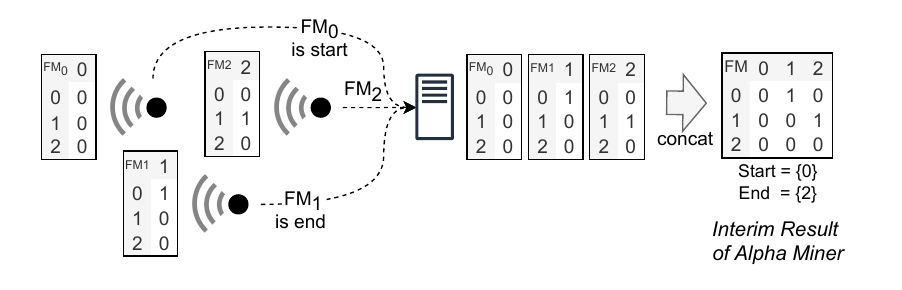}
    \caption{Phase 2 -- Requesting a Footprint Matrix: We request partial FMs and start/end activity flags from all nodes. 
    Upon receiving the data, we concatenate the matrices, form start and end activity sets, and compute the footprint matrix.}
    \label{fig:edgeminerphase2}
    \Description[Phase 2 -- Requesting a Footprint Matrix]{We request partial FMs and start/end activity flags from all nodes. 
    Upon receiving the data, we concatenate the matrices, form start and end activity sets, and compute the footprint matrix.}
\end{figure}

\subsection{Optimzaitons}
Next, we introduce three optimizations: Most-Frequent-Predecessors (MFPs), batching and a sliding window to enhance efficiency.

\subsubsection{Most-Frequent-Predecessors (MFPs)}
In Phase 1 of \project, a node -- upon detecting an event -- sends out requests to all other nodes to determine the event's predecessor. 
This results in $n-1$ requests where $n$ denotes the number of distinct nodes, i.e., activities in the event log.
To reduce communication overhead, \project contacts nodes in order of their likelihood of being the predecessor, i.e., it begins with the node that in the past was the most frequent predecessor according to its local FM. 
We show in a dataset analysis in Section~\ref{sec:Evaluation} that activities of real-world processes statistically only have a few predecessor activities, which further only rarely change throughout the lifetime of a process.
For example, even for datasets with a substantial number of activities, such as the Hospital Log~\cite{hospital} comprising more than 600 distinct activities, the average count of predecessors per activity remains modest, averaging fewer than 7 predecessors.
As a result, for each event, a node can very efficiently determine the predecessor, reducing the communication overhead by up to 96\% compared to querying all nodes in the network and thereby increasing scalability.
Additionally, our findings demonstrate that even datasets with numerous activities, such as the Hospital Log, stabilize in the number of requested nodes after only a relative small number of events.

\subsubsection{Batching}
To further minimize communication overhead in \project, we batch querying of predecessor events.
Rather than sending individual requests for each event, nodes accumulate a number of events before dispatching them in a single message.
The nodes continue to use \optimization.
As soon as the node finds the predecessor of an event, it removes the event from the batch. 
The node ceases to query additional nodes once the batch is empty.
For example, by batching predecessor queries in batches of 40 events, we show that \project can reduce the average number of queried nodes per event to less than 2.5\% of all nodes.

\subsubsection{Sliding Window}
To reduce resource demands, \project utilizes a sliding window:
Once the local event log reaches a given memory size or information age, the node discards it. 

\subsection{Correctness}
\label{subsubsec:correctness}

After introducing the algorithmic design of \project, we next show its correctness. 
We claim that merging the partial FMs in \project outputs the same FM as when calculating the FM based on a full event log.
Here, we discuss the intuition of the proof, a formal proof is in Appendix \ref{sec:Appendix}.

By definition of the FM, there is one column and row per activity. 
An entry of the FM describes how often the activity of the row was followed by the activity of the column.
A column of an FM therefore contains all predecessor relations of a specific activity.
In \project, the calculation of these columns is distributed to the nodes, and they are updated whenever the predecessor of an activity is found.
Every predecessor relation of a node, and therefore activity, is found by the design of \project and no direct successions are added to the FMs of multiple nodes but exactly one.
Thus, locally created columns in concatenated form ordered by activity ID must correspond exactly to the FM that is obtained when the entire log is considered and all direct successions are entered into a large matrix.

In our experimental evaluation in Section~\ref{sec:Evaluation}, we show (a) that both \project and established central algorithms result in identical FMs and thereby process models and (b) that in four out of five datasets an intermediate FM already complies with the FM of the whole event log by more than 90\% after less than 5\% of the events of the datasets.

\subsection{Distributed Mining with \project}
By distributing the computation of FMs, we can efficiently implement key process mining algorithms, such as, for example, the Alpha Miner, Heuristics Miner, and Inductive Miner, while also supporting online conformance checking in a scalable and privacy-preserving manner. 
Below, we outline how \project distributes each of these algorithms.

\subsubsection{Distributed Alpha Miner}
The Alpha Miner relies on identifying direct successions in event logs (cf. Section~\ref{sec:Background}). Hence, it builds on causal relations between activities. 
We leverage \project to distribute the computation of the FM from which the Alpha Miner constructs the process model. The Alpha Miner expects an FM consisting only of ones and zeros, i.e., if there is  directly-follows relation or not. As \project counts the number of successions, we simply set FM entries greater than zero to one.

\subsubsection{Distributed Heuristics Miner}
In contrast to the Alpha Miner, the Heuristics Miner not only considers whether a direct succession is present, but also considers the frequency and dependency measure of activity relationships. 
With a request, \project merges all partial FMs to a global FM at the requesting entity, which includes the frequencies of the direct successions. 
From these, it calculates the dependency measures of the activity relationships.
Next, the Heuristics Miner constructs a process model based on these frequencies and dependency measures. 

\subsubsection{Distributed Inductive Miner}
The Inductive Miner first constructs a directly-follows graph and then recursively splits it into subprocesses to derive a process tree as process model. 
To obtain the directly-follows graph the Inductive Miner needs the direct successions present in the event log as well as the set of start activities and the set of end activities.
\project offers the direct successions in the form of an FM along with the two sets.
Just as for the Alpha Miner, the frequencies of the direct successions are not relevant, so that all entries greater than zero are set to one.

\subsubsection{Distributed Footprint-Based Online Conformance Checking}
Online conformance checking continuously monitors process executions and compares them against a reference model in real time. Using \project, each node generates and maintains a local, partial FM that captures the current state of the process at that node. 
As new events are detected, the local FMs are updated, and each node can use its FM to perform conformance checking against a predefined process model.

\section{Evaluation}
\label{sec:Evaluation}

In this section, we evaluate \project in terms of correctness, accuracy, and performance.
We begin by discussing implementation, baselines and datasets in Section~\ref{subsec:Setup}.
Next, Section~\ref{subsec:datasetAnalysis} analyzes these datasets.
In Section~\ref{subsec:ExperimentalCorrectness}, we experimentally evaluate the correctness of \project. 
We examine how \optimization reduces the number of queried nodes for Phase 1 of the algorithm compared to naive querying in Section~\ref{subsec:queriednodes}.
Section~\ref{subsec:FMsizes} estimates FM sizes during an algorithm run, and Section~\ref{subsec:discussion} reflects on the evaluation results.

\subsection{Evaluation Setup}
\label{subsec:Setup}
\subsubsection{Implementation}
We implement \project in Python and use a network of docker containers to simulate the distributed nature of \project. 

\subsubsection{Baselines}
\label{subsec:baselines}

In our experiments, we compare \project with and without \optimization to two baselines:
(1) "Query All" represents the upper bound where all nodes are queried, while (2) a theoretical oracle-like approach serves as the lower bound.
The latter requires no requests if the predecessor corresponds to the node itself, and only one request to identify the predecessor otherwise.
Further, we compare \project to traditional process discovery algorithms \cite{alpha_miner,heuristicMiner,leemans2013discovering} to validate that the resulting FMs are identical.

\subsubsection{Datasets}
\label{subsec:data_sets}
We conduct experiments on the five rel-world datasets (event logs): Hospital Log~\cite{hospital}, Sepsis Cases~\cite{sepsis} and an IoT Event Log for Process Mining in Smart Factories\footnote{We use the \texttt{MainProcess.xes} file of the non-cleaned version.}~\cite{smart_factories}.
These process-mining event logs naturally support geographical distribution over multiple nodes.
To further substantiate our evaluation, we also include the BPI Challenge 2017~\cite{bpic17} and the Road Traffic Fine Management Process~\cite{traffic_fine} datasets. 
In our setting, each activity in a dataset represents a node in an IoT network.

\begin{table*}[tb]
    \centering
    \rowcolors{2}{white}{gray!10}
        \caption{Dataset properties and \project performance (Standard deviation in parentheses if not stated otherwise). 
        } 
        \begin{tabular}{lrrrrrrrrr}
            \textbf{Event\,Log}  & \textbf{Events} & \textbf{Activities\,(Start)} & \textbf{Cases} & \textbf{$\varnothing$Length} & \textbf{Self-Loops} & \textbf{$\varnothing$Predecessors} &  \textbf{$\varnothing$MFP\,Ratio} & \textbf{Reduction}\\
            Traffic Fine & 561,470 & 11 (1) & 150,370 & 3.73 &  4,306 & 6.36 (2.71) & 0.80 (0.17) & 64.94\% \\
            Sepsis & 15,214 & 16 (6) &1,050 & 14.49 & 1,034 & 7.19 (2.67) & 0.53 (0.20) & 81.30\% \\
            Smart Factories & 8,607 & 21 (1) & 271 & 31.76  & 5,854 & 4.62 (3.33)  & 0.69 (0.07) &  92.15\% \\
            %
            %
            BPIC 2017 &  1,202,267 & 26 (1) & 31,509  & 38.16 & 448,205  & 6.85 (5.30) & 0.79 (0.18)  & 93.40\% \\
            Hospital & 150,291 & 624 (29) & 1,143 & 131.49  & 26,542 & 6.78 (12.25) & 0.69 (0.28) & 96.13\% \\
        \end{tabular}
        \label{table:analyzedDataSetsPredecessor}
\end{table*}

\subsection{Dataset Analysis}
\label{subsec:datasetAnalysis}
We begin our experimental evaluation with an analysis of our real-world event logs, see Table~\ref{table:analyzedDataSetsPredecessor}.
Our analysis shows that even in large datasets like the Hospital Log with over 600 distinct activities, the average number of preceding activities per activity stays below 7, showcasing the potential of \optimization to reduce the average number of queried nodes per event.
Further, our analysis reveals that event logs like the Smart Factories dataset contain numerous self-loops, resulting from different lifecycle transitions and states within the recorded events.
This characteristic benefits \project, as self-loops do not require communication.

\subsection{Experimental Correctness \& Fitness}
\label{subsec:ExperimentalCorrectness}

In addition to the analytical correctness proof in Section~\ref{subsubsec:correctness}, we experimentally validate that the resulting FMs of \project are identical to the ones computed by today's centralized algorithms\footnote{we compare implementations in \texttt{pm4py} (https://pm4py.fit.fraunhofer.de)}.
We test both on randomly generated artificial datasets and common real-world event logs, see Section \ref{subsec:data_sets}.
For all datasets the resulting FMs are identical at the end of the log, see Fig.~\ref{fig:plot_fitness}. 

Next, we examine how the merged FMs converge over time:
Fig.~\ref{fig:plot_fitness} demonstrates that nearly all datasets reach 90\% fitness within less than 5\% of the trace
Here, the Smart Factories dataset is an exception, reaching 90\% fitness after processing approximately 41\% of the events, as it is quite small, see Table~\ref{table:analyzedDataSetsPredecessor}.
The results indicate that querying interim results can be highly informative, potentially eliminating the need to process all events 
before requesting the first FM. 
Moreover, our analysis suggests that reducing the sliding window for discarding old data has minimal impact on the resulting FM.

\subsection{Number of Queried Nodes}
\label{subsec:queriednodes}

To assess \project's performance, we analyze its two phases separately.
For the first experiments, we assume a batch size of one.

\begin{figure}[tb]
    \centering
    \includegraphics[width=\linewidth]{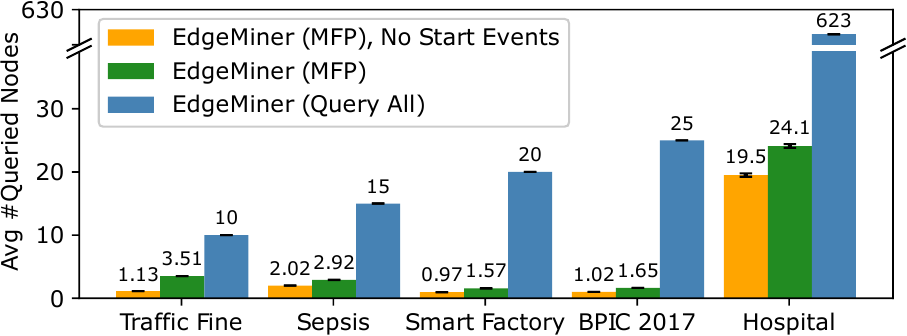}
    \caption{Average number of nodes queried with and without \optimization, including standard error. 
    \optimization and knowledge of the start events reduce communication demands by a factor of 7.5 to 30 depending on the dataset. 
    }
    \label{fig:bar_all}
    \Description[Comparison of average queried nodes with \optimization and querying all nodes]{Average number of nodes queried with and without \optimization, including standard error. 
    \optimization and knowledge of the start events reduce communication demands by a factor of 7.5 to 30 depending on the dataset.}
\end{figure}

\begin{figure*}[tb]
    \begin{minipage}[t]{0.32\textwidth}
        \centering
        \includegraphics[width=\linewidth, trim=0cm 0.3cm 0cm 0cm]{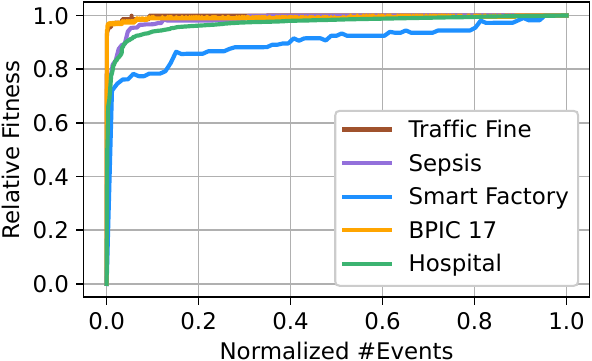}
        \caption{Fitness over time. 
        Intermediate FMs quickly converge to the centralized computed FM.
        The BPIC 2017 dataset, for example, already has a fitness of over 90\% after 200 events.
        }
        \label{fig:plot_fitness}
        \Description[Fitness over time.]{Intermediate FMs quickly converge to the centralized computed FM.
        The BPIC 2017 data set, for example, already has a fitness of over 90\% after 200 events.
        }
    \end{minipage}%
    \hfill
    \begin{minipage}[t]{0.32\textwidth}
        \includegraphics[width=1\linewidth, trim=0cm 0.3cm 0cm 0cm]{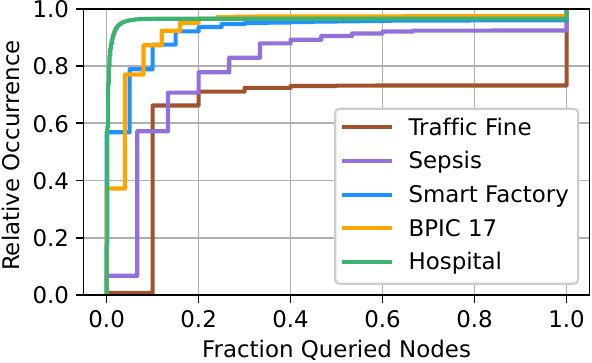}
        \caption{Cumulative Distribution Functions: 
        A steeper curve indicates a higher proportion of events that queried few nodes to identify their predecessors, reflecting greater efficiency.
        }
        \Description[Cumulative Distribution Functions]{Cumulative Distribution Functions: 
        A steeper curve indicates a higher proportion of events that queried few nodes to identify their predecessors, reflecting greater efficiency.}
        \label{fig:cdf_all}
    \end{minipage}%
    %
    %
    \hfill
    \begin{minipage}[t]{.32\textwidth}
        \includegraphics[width=1\linewidth, trim=0cm 0.3cm 0cm 0cm ]{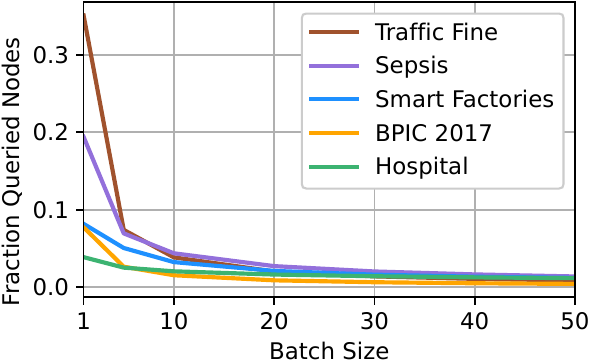}
        \caption{Relationship between batch size and average number of queried nodes per event, normalized. 
        Batching efficiently reduces the number of queries. 
        }
        \label{fig:batching}
        \Description[Relationship between batch size and average number of queried nodes per event]{Relationship between batch size and average number of queried nodes per event, normalized. 
        Batching efficiently reduces the number of quries. 
        }
    \end{minipage}
\end{figure*}

\begin{figure}[tb]
    \centering
    \includegraphics[width=1\linewidth, trim=0cm 0.3cm 0cm 0cm]{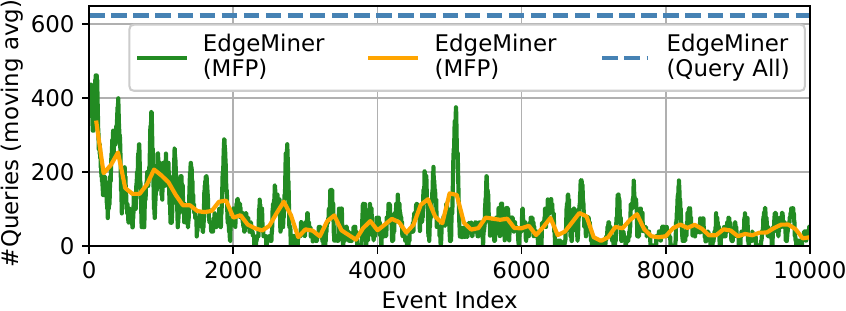}
    \caption{Number of nodes queried over time in the Hospital Log (moving average; green line: window size 50, step 1; orange line: window size 200, step 100). 
    The moving average stabilizes after around 3,000 events.
    A reference line (Query All) indicates the maximum number of requested nodes per event.
    Note that the plots are limited to the first 10,000 of over 150,000 events.
    }
    \label{fig:moving_window_hospital}
    \Description[Number of nodes queried over time in the Hospital Log]{Number of nodes queried over time in the Hospital Log (moving average; green line: window size 50, step 1; orange line: window size 200, step 100). 
    The moving average stabilizes after around 3,000 events.
    A reference line (Query All) indicates the maximum number of requested nodes per event.
    Note that the plots are limited to the first 10,000 of over 150,000 events.}
\end{figure}

\begin{figure}[tb]
    \centering
    \includegraphics[width=\linewidth]{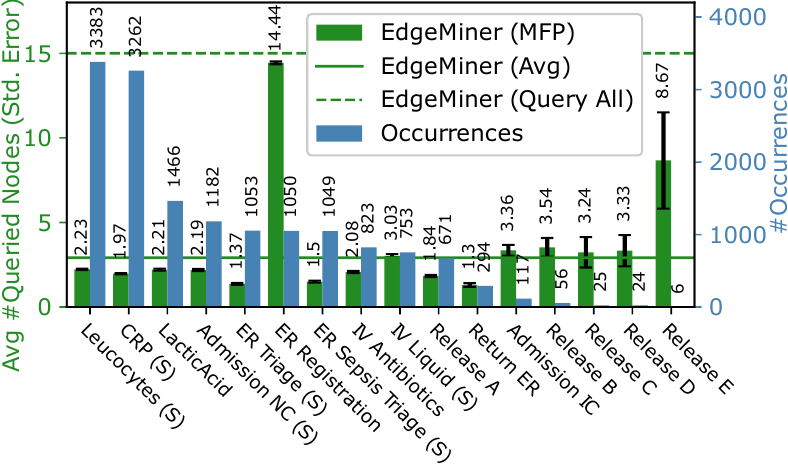}
    \caption{Bar diagram based on the Sepsis Cases illustrating the average number of queried nodes per activity (left y-axis) alongside the number of occurrences of each activity (right y-axis). 
    The x-axis displays the different activities of the event log and highlights the start activities (S).
    We plot the number of requested nodes when querying all nodes and the average number of queried nodes over all events using \optimization. 
    The diagram shows that activities with the highest average queries occur the least, excluding start activities.
    }
    \label{fig:bar_sepsis}
    \Description[Bar diagram based on Sepsis Cases]{Bar diagram based on the Sepsis Cases illustrating the average number of queried nodes per activity (left y-axis) alongside the number of occurrences of each activity (right y-axis). 
    The x-axis displays the different activities of the event log and highlights the start activities (S).
    The diagram shows that activities with the highest average queries occur the least, excluding start activities.}
\end{figure}

\subsubsection{Comparison Across Datasets}
Fig.~\ref{fig:bar_all} and Table~\ref{table:analyzedDataSetsPredecessor} illustrate, for each dataset, the average number of nodes queried per event to find the predecessor. 
Query All serves as the naive baseline of messaging every node in the network. 
With \optimization, nodes require significantly fewer requests to identify their predecessor event.
We achieve additional message reduction when the process has a designated start event, making predecessor requests for the respective node redundant.
Overall, the reduction increases with the number of activities, influenced by factors such as an expanding average case length, diminishing significance of start events, varying numbers of self-loops relative to total events, and the presence of just a few highly frequent predecessors.

%
Fig.~\ref{fig:cdf_all} shows Cumulative Distribution Functions (CDF) for all datasets. 
The Hospital Log has the steepest increase, with over 90\% of events needing queries to no more than 5\% of the nodes. 
Conversely, the Traffic Fine dataset has very few instances where the predecessor can be identified after querying just 5\% of the nodes.
This deviation stems from the limited number of 11 activities in the Traffic Fine dataset. 
Hence, within the 5\% threshold, only self-loops are considered.
\subsubsection{Stabilization over Time}
We now delve deeper into the specifics of the Hospital Log, as it is by far the largest one in terms of activities. 
Fig.~\ref{fig:moving_window_hospital} illustrates how the number of queried nodes stabilizes in the Hospital Log after fewer than 3,000 events.
There are still fluctuations, for example, at around 5,000 events. 
This event log is a good example of stabilization due to its large activity set.
Datasets with more activities take longer to stabilize, as each activity must occur multiple times for the MFPs to be identified.

\subsubsection{Activity-wise Analysis}
%

Next, we evaluate the number of queries per activity in the Sepsis Cases log, see Fig.~\ref{fig:bar_sepsis}.
In contrast to other event logs, the Sepsis Cases do not have one designated start event but six different ones.
However, about 95\% of cases start with "ER Registration", which explains why the corresponding bar clearly stands out above the others.
Although "Release~E" shows an increased query count, it is not a start activity. 
The increased query count results from the limited occurrence (only six times), indicating non converged predecessor stabilization.
The situation is similar for the four activities to its left.
They occur infrequently, which also explains the higher standard error.

\subsubsection{Batching}
\label{subsubsec:batching}
Fig.~\ref{fig:batching} illustrates the reduction in the average number of queried nodes per event as the batch size increases.
For example, at a batch size of 40, all datasets achieve an average number of nodes queried per event of less than 2.5\% of the nodes in the network.
For the Traffic Fine datasets, 2.5\% are on average 0.25 nodes, whereas for the Hospital Log, those are on average 15.6 nodes.

\subsubsection{Communication Cost and Privacy}

Table~\ref{table:communication_vs_privacy} compares the centralized approach with the naive \project algorithm and \project using \optimization and batching in terms of communication overhead and privacy.
The centralized approach requires fewer messages per event, but at the cost of reduced privacy.
For the naive \project, the number of messages per event is equal to $2\cdot(n-1)+1$ with $n$ the number of activities/nodes. 
All nodes except the node itself are contacted, all other nodes reply, and the chosen node gets a confirmation message.
For \project with \optimization and batching we compute $2\cdot k + 1$ with $k$ the average number of queried nodes per event (see Fig.~\ref{fig:bar_all}).
In the centralized approach, event data is sent to a central entity right after the event, which requires only one message per event. 
However, this approach inherently relies on a central storage system where all events are collected, presenting privacy and scalability risks. 
Although the number of messages per event for \project with \optimization and batching is higher than in the centralized approach, it is a considerable improvement compared to the naive \project.
We use a batch size of 10 events.
As we show in Section~\ref{subsubsec:batching}, batching with larger batch sizes reduces the number of queried nodes and therefore the number of messages per event even further.
The trade-off presented here is clear: while the \project variants, particularly \project with \optimization and batching, cause higher communication overhead, they offer substantial scalability and privacy benefits by eliminating central storage and limiting non-aggregated communication to within the network.\\
\begin{table}[tb]
    \centering
    \rowcolors{2}{gray!10}{white}
        \caption{EdgeMiner's efficiency and privacy across datasets against Centralized PM (CS/CE = Central Storage/Entity).
        }
        \setlength{\tabcolsep}{3.5pt}
        \begin{tabularx}{\textwidth}{llccc}
            &   &    \textbf{Centr.} & \textbf{\thead[c]{Naive\\ \project}} & \textbf{\thead[c]{\project\\ with MFP \& \\Batching (10)}} \\
              & Traffic Fine & 1 & 21 & 1.76   \\
             \cellcolor{gray!10}     & Sepsis & 1 & 31 & 2.32  \\ 
               & Smart Factories & 1 & 41 & 2.30  \\ 
              \cellcolor{gray!10}   & BPIC 2017 &  1 & 51 & 1.76 \\
             \multirowcell{-5}{\textbf{\# Msg.}\\\cellcolor{gray!10}\textbf{per Event}}   & Hospital & 1 & 1247 & 26.46 \\ 
             \cellcolor{white} & Storage & None & No CS & No CS \\
            \cellcolor{white}\multirowcell{-2}{\cellcolor{white}\textbf{Privacy}\\\cellcolor{white}\textbf{Enablers}} & Msg. Scope & To CE & In-Network & In-Network
        \end{tabularx}
        \label{table:communication_vs_privacy}
\end{table}

\subsection{Footprint Matrix Sizes}
\label{subsec:FMsizes}

In the second phase of \project, nodes transmit their data to a (central) entity, generating a fixed number of messages.
These messages include the node's local FM along with two additional Booleans indicating whether a node recorded the start or end activity of a trace.
In \project, each node only keeps track of their predecessors in its FM, resulting in transmitting FMs across the network of a maximal size of $n \times 1$, where $n$ represents the total number of activities.
Even for datasets like the Hospital Log, which has an activity count of 624, each node transmits only a list consisting of at most 624 elements. 
Compared to existing distributed process discovery algorithms~\cite{alpha_miner_map_reduce,alpha_miner_distributed}, which often require sending up to quadratic FMs to a central entity, our approach is characterized by a compact message size.
Moreover, storing such matrix sizes on IoT nodes is feasible even with constrained resources. 
For long-running processes, \project can further discard event data related to own, predecessor, and successor events once a node's partial FM has been collected.

\subsection{Evaluation Discussion}
\label{subsec:discussion}

\project leverages data locality to enhance its efficiency. 
Unlike existing distributed process discovery algorithms~\cite{alpha_miner_map_reduce,alpha_miner_distributed}, which typically collect the event data completely before processing it and creating FMs or process models, \project adopts a proactive strategy by processing data at runtime and at its source. 
When the FM is requested, \project has already started the process discovery algorithm, by collecting partial FMs at each node, and thereby speeding up mining. 
Consequently, our strategy provides a significant time advantage, especially when handling large datasets with hundreds of activities and thousands of events.
When looking at conformance checking, \project enables online conformance checking by distributing the computation of FMs and making the partial FMs available at the data sources.

\section{Conclusion}
\label{sec:Conclusion}

In this paper, we introduce \project, a novel approach tailored to enable distributed process discovery and conformance checking in IoT networks.
By distributing computations to IoT and edge devices, \project addresses scalability and privacy concerns inherent in centralized approaches. 
We discuss how \project's design principles support common process mining algorithms that  construct footprint matrices or record direct successions, such as the Alpha, Heuristics and Inductive Miner~\cite{alpha_miner,heuristicMiner,leemans2013discovering}.
Moreover, integrating Most-Frequent-Predecessor Requesting reduces communication overhead, ensuring efficient event handling within the network. 
The batching of predecessor queries further minimizes communication overhead.

Our empirical evaluations showcase \project's stark reduction of messages in real-world scenarios compared to a naive predecessor determination approach.
\project achieves a reduction in communication for determining predecessor events by up to 96\% compared to querying all network nodes after each event.
By batching predecessor queries in groups of, for example, 40 events, experiments show an average reduction in the number of queried nodes per event to less than 2.5\% of all nodes.
We analytically and experimentally demonstrate the correctness of the algorithm and show how the intermediate results of \project quickly converge to the FM of the full event log. 
We highlight how \project provides a significant time advantage over algorithms lacking pre-processing capabilities while still producing the same FMs as the centralized algorithms.
By avoiding the storage of full event traces and sensitive personal data in central locations, \project enables data privacy through aggregation that eliminates personal identification.

\begin{acks}
  This work received funding by the Deutsche Forschungsgemeinschaft (DFG, German Research Foundation), grant 496119880.
\end{acks}

\bibliographystyle{ACM-Reference-Format}
\bibliography{edgeminer}

\appendix
\section{Appendix: Formal Correctness Proof}
\label{sec:Appendix}

We introduce a formal correctness proof. 
For this, we show that the merging of partial FMs in \project outputs the same FM as calculating the FM based on a full event log.

We begin by introducing notations before defining our claim. 
Let $FM^{all}$ denote the FM computed from an entire event log $L$ and $n := \vert T_L \vert$, where $T_L$ denotes the set of all activities in $L$.
Further, let $FM^i$ for $i\in \{0,\dots, n-1\}$ be the partial FM computed by node $i$, and $L_i$ the partial event log that node $i$ has knowledge of.
Since $i$ only keeps track of its predecessor events, its footprint matrix $FM^i$ contains only one column. 
In the following correctness claim, we only consider the footprint matrices and omit start and end activities. 
The design of \project ensures that a node considers an activity only a start activity or end activity if it corresponds to the start or end event of a full trace.
Further, we assume that the footprint matrices only consist of ones and zeroes. 
The correctness for footprint matrices with entries in $\mathbb{N}_0$ follows from the fact that \project detects all direct successions and adds a direct succession to exactly one partial footprint matrix.

\begin{claim}
    Merging the partial footprint matrices $FM^0, \dots, FM^{n-1}$ by concatenating them ordered by activity ID results in the same FM as the $FM^{all}$, formally
    $\begin{pmatrix} FM^0 &~ \dots ~& FM^{n-1} \end{pmatrix} ~=~ FM^{all}.$
\end{claim}

\begin{proof}
    To prove the claim, we show that all direct successions represented in the partial FMs can be found in $FM^{all}$ and vice versa.
    Thus proving the claim of equivalence between the merged partial FMs and the FM with unrestricted knowledge.

    By definition, $FM^{all} \in \{0,1\}^{n\times n}$.
    As already discussed, for $FM^i, ~i\in \{0,\dots, n-1\}$, we assume without loss of generality that $FM^i \in \{0,1\}^{n\times 1}$.
    It directly follows $\begin{pmatrix} FM^0 &~ \dots ~& FM^{n-1} \end{pmatrix} \in \{0,1\}^{n \times n}$.
    
    Let $i >_{L_j} j, ~i,j \in \{0, \dots, n-1\}$ be a direct succession encoded in the merged partial footprint matrix $\begin{pmatrix} FM^0 &~ \dots ~& FM^{n-1} \end{pmatrix}$. 
    Due to the concatenation of the partial matrices, it holds $FM^j_i = 1$.
    $FM^j$ contains all direct successions between its activity and its preceding activities.
    This means there is a trace in the event log $L$ in which $i$ is followed by $j$. 
    By definition, $FM^{all}$ captures all direct successions in $L$, meaning it must hold $FM^{all}_{ij} = 1$.

    It remains to show that the opposite direction holds.
    Let $k>_L \ell,~k, \ell \in \{0, \dots, n-1\}$ be a direct succession represented in $FM^{all}$, formally $FM^{all}_{k\ell} = 1$.
    Since each node keeps track of all its predecessor activities in its local FM, we know that $FM^\ell_{k} = 1$.
    Merging all partial FMs only means concatenating them ordered by activity ID.
    Therefore, it follows $\begin{pmatrix} FM^0 &~ \dots ~& FM^{n-1} \end{pmatrix}_{k\ell} = 1$.
    Thus proving the claim.
\end{proof}